\documentclass[a4paper,11pt]{article}
\usepackage[margin=2cm]{geometry}

\usepackage[utf8]{inputenc}
\usepackage{amsmath}
\usepackage{amssymb}
\usepackage{amsthm}
\usepackage{graphicx}
\usepackage{float}
\usepackage{xcolor}
\usepackage{cite}
\usepackage{microtype}
\usepackage{hyperref}

\newtheorem{theorem}{Theorem}
\newtheorem{lemma}{Lemma}

\newtheorem{remark}{Remark}

\usepackage[ruled,vlined,linesnumbered,noend]{algorithm2e}

\title{Finding single-source shortest $p$-disjoint paths:\\ fast computation and sparse preservers}

\author{
  Davide~Bilò\\
  Department of Humanities and Social Sciences, University of Sassari, Italy\\
  \texttt{davide.bilo@uniss.it}
  \and
  Gianlorenzo~D'Angelo\\
  Gran Sasso Science Institute, L'Aquila, Italy\\
  \texttt{gianlorenzo.dangelo@gssi.it}
  \and
  Luciano~Gualà\\
  Department of Enterprise Engineering, University of Rome ``Tor Vergata'', Italy\\
  \texttt{guala@mat.uniroma2.it}
  \and
  Stefano~Leucci\\
  Department of Information Engineering, Computer Science and Mathematics,\\ University of L'Aquila, Italy\\
  \texttt{stefano.leucci@univaq.it}
  \and
  Guido~Proietti\\
  Department of Information Engineering, Computer Science and Mathematics,\\ University of L'Aquila, Italy\\
  Institute for System Analysis and Computer Science ``Antonio Ruberti'', Italy\\
  \texttt{guido.proietti@univaq.it}
  \and
  Mirko~Rossi\\
  Gran Sasso Science Institute, L'Aquila, Italy\\
  \texttt{mirko.rossi@gssi.it}
}

\begin{document}

\maketitle

\begin{abstract}
Let $G$ be a directed graph with $n$ vertices, $m$ edges, and non-negative edge costs.
Given  $G$, a fixed source vertex $s$, and a positive integer $p$, we consider the problem of computing, for each vertex $t\neq s$, $p$ edge-disjoint paths of minimum total cost from $s$ to $t$ in $G$.
Suurballe and Tarjan~[Networks, 1984] solved the above problem for $p=2$ by designing a $O(m+n\log n)$ time algorithm which also computes a sparse \emph{single-source $2$-multipath preserver}, i.e., a subgraph
containing $2$ edge-disjoint paths of minimum total cost from $s$ to every other vertex of $G$.  The case $p \geq 3$ was left as an open problem.

We study the general problem ($p\geq 2$) and prove that any graph admits a sparse single-source $p$-multipath preserver with $p(n-1)$ edges. This size is optimal since the in-degree of each non-root vertex $v$ must be at least $p$.
Moreover, we design an algorithm that requires $O(pn^2 (p + \log n))$ time to compute both $p$ edge-disjoint paths of minimum total cost from the source to all other vertices and an optimal-size single-source $p$-multipath preserver.
The running time of our algorithm outperforms that of a natural approach that solves $n-1$ single-pair instances using the well-known
\emph{successive shortest paths} algorithm by a factor of $\Theta(\frac{m}{np})$ and is asymptotically near optimal if $p=O(1)$ and $m=\Theta(n^2)$.
Our results extend naturally to the case of $p$ vertex-disjoint paths.
\end{abstract}
\clearpage
 
\section{Introduction}

Consider a communication network modelled as a directed graph $G$ with $n$ vertices, $m$ edges, and non-negative edge costs. 
Whenever a \emph{source} vertex $s$ needs to send a message to a \emph{target} vertex $t$, we are faced with the problem of finding a \emph{good} path connecting $s$ and $t$ in $G$.
Typically, this path is chosen with the aim of minimizing the communication cost, i.e., the sum of the costs of the path's edges.
In a scenario where some edges of the network might be congested or faulty, it is useful to introduce some degree of redundancy in order to improve the communication reliability.
One of the possible approaches that aims to formalize the above requirements asks to find $p$ edge-disjoint paths from $s$ to $t$ in $G$ for some integer $p \ge 2$. Quite naturally, similarly to the case of the shortest path, we would like to minimize the sum of the costs of the edges in the selected paths.

This is equivalent to the problem of computing a minimum-cost flow of value $p$ from $s$ to $t$ in the unit-capacity network $G$ and can be solved in time $O(p(m+n\log n))$ using the \emph{successive shortest paths} (SSP) algorithm \cite{network_flows,erickson_mincost_flow}.

In this paper we focus on the \emph{single-source} case, in which a fixed source vertex $s$ wants to communicate with
every other vertex $t$ using $p$ edge-disjoint paths.
We distill the above discussion into the following two problems:
\begin{description}
    \item[Single-source $p$-multipath preserver problem:] We want to find a \emph{sparse} subgraph $H$ of $G$ such that, for every vertex $t \neq s$, $H$ contains $p$ edge-disjoint paths of minimum total cost from $s$ to $t$ in $G$. We will refer to such a subgraph $H$ as a \emph{single-source $p$-multipath preserver}.
    Among all possible feasible solutions, we aim at computing the one of minimum  \emph{size}, i.e., having the minimum number of edges.    
    \item[Shortest $p$ edge-disjoint paths problem:] For every vertex $t \neq s$, we want to compute a subset $S^t$ of edges from $G$ that induce $p$ edge-disjoint paths of minimum total cost from $s$ to $t$ in $G$.
\end{description}

Observe that the if the graph $G$ is not sufficiently connected, the single-source $p$-multipath preserver $H$ and some of the sets $S^t$ defined above might not exist.
To avoid this issue, we assume that $G$ is $p$-edge-outconnected from $s$, i.e., given any vertex $t \neq s$, $G$ contains $p$ edge-disjoint paths from $s$
to $t$.\footnote{It is possible to check whether a graph is $p$-edge-outconnected from $s$ in $O(pm \log \frac{n^2}{m})$ time \cite{Gabow95}.}

The above problems have been addressed by Suurballe and Tarjan for the special case $p=2$ in \cite{DBLP:journals/networks/SuurballeT84}, where they provide an algorithm requiring time $O(m + n \log n)$ to compute both a single-source $p$-multipath preserver of size  $2(n-1)$ and (a compact representation of) all sets $S^t$ of the shortest $p$ edge-disjoint paths problem.\footnote{After the execution of their algorithm, it is possible to compute each set $S^t$ in time $O(|S^t|)$.}
In their paper, the authors mention the case $p > 2$ as an important open problem.

\noindent In this paper we provide the following results:
\begin{itemize}
    \item We prove that any graph $G$ always admits a single-source $p$-multipath preserver of size $p(n-1)$. This size is optimal since the in-degree of each non-source vertex $t$ in $H$ needs to be at least $p$, even to preserve the $p$-edge-outconnectivity from $s$ to $t$. 
    \item We design an algorithm that requires $O(pn^2 (p + \log n))$ time to solve the shortest $p$ edge-disjoint paths problem. This improves over the natural algorithm that computes the sets $S^t$ with $n-1$ independent invocations of the SSP algorithm, which would require $O(pnm+ pn^2\log n)$ time. Up to logarithm factors, our algorithm is $\Theta(\frac{m}{np})$ times faster than the above algorithm based on SSP. Moreover, for $p = O(1)$ and $m = \Theta(n^2)$, the time complexity of our algorithm is optimal up to logarithmic factors. 
    Finally, our algorithm also computes a single-source $p$-multipath preserver $H$ of optimal size that contains all the edges in the sets $S^t$.
\end{itemize}

We point out that a modification of our algorithm allows us to handle graphs $G$ that are not $p$-edge-outconnected from $s$.
In this case, our algorithm computes, for each vertex $t \neq s$, a set of edges $S^t$ that induce $\sigma(t)$ edge-disjoint paths from $s$ to $t$ of minimum total cost in $G$, where $\sigma(t)$ is the minimum between $p$ and the maximum number of edge-disjoint paths from $s$ to $t$.
Moreover, the algorithm also returns a subgraph $H$ of $G$ of optimal size where each vertex $t \neq s$ has exactly $\sigma(t)$ incoming edges. As in the previous case, $H$ contains all edges in the sets $S^t$. The running time of our algorithm is asymptotically unaffected.

We also discuss a variant of the problem in which, instead of minimizing the overall cost of $p$ edge-disjoint paths, we aim at minimizing the cost of the path with maximum cost. We show that our algorithm provides an optimal $p$-approximation, unless $\mathrm{P}=\mathrm{NP}$.

Finally, our results can be extended to the case of vertex-disjoint paths via a standard transformation of the input graph \cite{DBLP:journals/networks/SuurballeT84}. 
All the above modifications and variants are discussed in Section~\ref{sec:extensions}.
Some of the proofs are deferred to the Appendix.

\subparagraph*{Related work.}
As already mentioned, the closest related work is the paper by Suurballe and Tarjan that studies the case $p=2$~\cite{DBLP:journals/networks/SuurballeT84}.

The single-source $p$-multipath preserver problem falls within a broad class of problems 
with a long-standing research tradition.
Here we are given a graph $G$ and we want to select a sparse subgraph $H$ of $G$ which maintains, either in an exact or in an approximate sense, some distance-related property of interest.
The goal is that of understanding the best trade-offs that can be attained between the size of $H$ and the accuracy of the maintained properties.
As a concrete example, if we focus on the cost of a single path (i.e., $p=1$) between pairs of vertices, a well-known notion adopted is that of \emph{graph spanners}, which has been introduced by Peleg and Sch\"{a}ffer~\cite{PelegS89}. A spanning subgraph $H$ is an $\alpha$-spanner of $G$ if the distance of each pair of vertices in $H$ is at most $\alpha$ times the corresponding distance in $G$. If $G$ is undirected then it is possible to compute, for any integer $k \ge 1$, a $(2k-1)$-spanner of size $O(n^{1+\frac{1}{k}})$ \cite{AlthoferDDJS93} (if we assume the Erd\H{o}s Girth Conjecture \cite{erd1965extremal}, this trade-off is asymptotically optimal), while there exist directed graphs for which any $\alpha$-spanner has size $\Omega(n^2)$.

When $\alpha=1$ and hence $H$ retains the exact distances of $G$, a $1$-spanner is usually called a \emph{preserver}. 
While $\Omega(n^2)$ edges might be necessary to preserve all-to-all distances, better trade-offs can be obtained if we only care to preserve distances between some pairs of vertices. For example, a shortest-path tree can be seen as a sparse \emph{single-source} preserver. More significant trade-offs can be obtained for different choices of the pairs of interest (see, e.g.,~\cite{Bodwin17,BodwinW16}). For more related results on the vast area of spanners and preservers, we refer the interested reader to the survey in \cite{spanner_survey}. 

Concerning the case of multiple paths ($p>1$), 
Gavoille et al.~\cite{GavoilleGV10} introduced the notion of
\emph{$p$-multipath spanner} of a weighted graph $G$, from which we borrow the term \emph{multipath}. A $p$-multipath $\alpha$-spanner of $G$ is a spanning subgraph $H$ of $G$ containing, for each pair of vertices $u,v$, $p$ edge-disjoint paths from $u$ to $v$ of total cost at most $\alpha$ times the cost of the cheapest $p$ edge-disjoint paths from $u$ to $v$ in $G$. Among other results, the authors of \cite{GavoilleGV10} prove the existence, for any choice of $p$, of a $p$-multipath $p(2k-1)$-spanner of size $O(p n^{1+\frac{1}{k}})$ for undirected graphs. 
Following \cite{GavoilleGV10}, there has been
further work on multipath spanners \cite{ChechikGP12,GavoilleGV11}. All of the above papers, however, focus on approximated costs, in the all-pairs setting on undirected graphs.
Since our focus is on $p$-multipath $\alpha$-spanners for directed graphs, in the single-source case, and for $\alpha=1$, such results cannot be directly compared to the one in this work. 

As discussed above, edge-disjoint paths can be seen as a strategy to achieve \emph{fault-tolerance} through redundancy.
Other approaches to address faults in networks, which aim at (approximately) preserving the length of the surviving shortest paths from a source vertex $s$, are captured by the notion of single-source fault-tolerant spanners and preservers \cite{BiloGLP18,ParterP16,ParterP18,Parter15,GuptaK17,BodwinGPW17,BaswanaCR18,BiloG0P16,BaswanaCHR20}.

\section{Preliminaries}
\label{sec:successive_shortest_path}
We denote by $V(G)$, $E(G)$, and $c:E(G)\rightarrow \mathbb{R}^+$, the set of vertices, the set of edges, and the cost function, respectively.
With a slight abuse of notation, if $S$ is a set of edges (resp. $\pi$ is a path), we denote by $c(S)$ (resp. $c(\pi)$) the sum of the costs $c(e)$ for $e \in S$ (resp. $e \in E(\pi)$).

In order to lighten the notation, in the rest of the paper we will assume that the graph is anti-symmetric, i.e., if $(u,v)\in E(G)$, then $(v,u)\not\in E(G)$. 
We make this assumption as we will define auxiliary graphs on the vertex set $V(G)$ in which some edge $(u,v)\in E(G)$ might appear in the reversed direction $(v,u)$ and  therefore, a non anti-symmetric graph $G$ may cause the presence of two parallel edges in the auxiliary graphs. It is easy to remove this assumption by distinguishing the two possible parallel edges with unique identifiers.

\subparagraph{Relation with the $s$-$t$-min-cost flow problem.}
\begin{figure}
    \centering
    \includegraphics[scale=1.1]{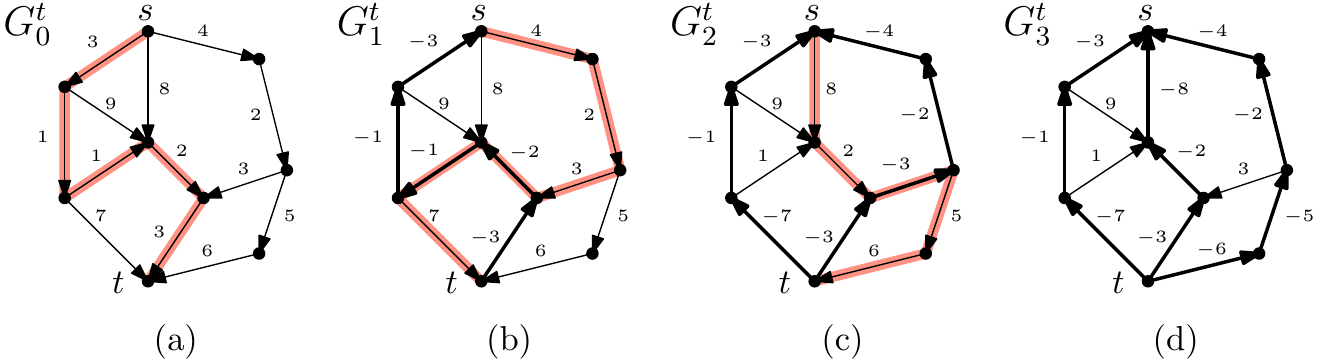}
    \caption{\small An execution of the successive shortest paths algorithm on a graph $G$. Figures (a), (b), (c), and (d) respectively show the residual networks $G_0^t = G$, $G_1^t$, $G_2^t$, and $G_3^t$. The shortest paths computed by the algorithm are highlighted in red. The edges that appear in the opposite orientation w.r.t.\ $G$ are shown in bold. If we orient the bold edges in $G_i^t$ as in $G$, we obtain the edges in $S_i^t$, i.e., those belonging to $i$ edge-disjoint paths of minimum total cost from $s$ to $t$ in $G$.}
    \label{fig:augment}
\end{figure}
For a fixed pair of vertices $s,t \in V(G)$, the problem of finding $p$ edge-disjoint paths of minimum total cost from $s$ to $t$  is a special case of the \emph{$s$-$t$-min-cost flow problem} where edges have unit capacities and the goal is to send $p$ units of flow from $s$ to $t$ at minimum total cost. 
\emph{Successive shortest path} (SSP) \cite{network_flows, erickson_mincost_flow} is a well-known algorithm that solves the $s$-$t$-min-cost flow problem. We now give a brief description of SSP for the special case of unit edge capacities.

The algorithm sends $p$ units of flow from $s$ to $t$ by iteratively pushing
one new unit of flow through a shortest path from $s$ to $t$ in the \emph{residual network} associated with the current flow.
More precisely, let the initial residual network be $G_0^t = G$. In the generic $i$-th iteration, SSP finds a shortest path $\pi_t$ from $s$ to $t$ in $G_{i-1}^t$, and uses it to compute a residual network $G_i^t$.

The residual network $G_i^t$ is obtained from $G^t_{i-1}$ by \emph{reversing} all the edges in $\pi_t$, where \emph{reversing} an edge $(u,v)$ of cost $c(u,v)$ means replacing $(u,v)$  with the edge $(v,u)$ of cost $c(v,u) = -c(u,v)$.
See Figure~\ref{fig:augment} for an example.

At the end of the $i$-th iteration, the $i$ units of flow are sent through the edges of $G$ that are reversed in the residual network $G_i^t$. We denote by $S_i^t$ the set of such edges, which contains exactly the edges of $i$ edge-disjoint paths from $s$ to $t$ of minimum total cost. Therefore, once the $p$-th iteration is completed, $S_p^t$ is a solution for the problem of finding $p$ edge-disjoint paths of minimum total cost from $s$ to $t$. An interesting observation that we will use later on is the following.
\begin{remark}\label{remark:oplus}
The set $S_i^t$ can be computed from $S_{i-1}^t$ and $\pi_t$ in $O(|S_i^t|+n)$ time by first setting $S_i^t=S_{i-1}^t$ and then by (i) deleting from $S_i^t$ all edges $(u,v) \in S_i^t$ that are reversed in $E(\pi_t)$, and (ii) adding to $S_i^t$ all edges $(u,v) \in E(\pi_t) \cap E(G)$.
\end{remark}

A straightforward implementation of the above algorithm requires time $O(pnm)$ since it computes $p$ shortest paths using the Bellman-Ford algorithm (notice, indeed, that the edge costs in the residual networks might be negative).
The above time complexity can be improved to $O(p(m+n\log n))$ by suitably re-weighting the residual network so that edge costs are non-negative and shortest paths are preserved, allowing the Dijkstra algorithm to be used in place of Bellman-Ford~\cite{erickson_mincost_flow}.

We can solve $n-1$ separated instances of $s$-$t$-min-cost flow (one for each node $t$) and obtain (i) the solution for the shortest $p$ edge-disjoint paths problem, i.e., the sets $S^t_p$ for each $t \in V(G) \setminus \{s\}$; (ii) a single-source $p$-multipath preserver by  making the union of all solutions $S^t_p$ obtained. However, the resulting single-source $p$-multipath preserver may not be sparse and the total running time needed to solve both problems is  $O(pn(m+n\log n))$. 

In Section $\ref{sec:size}$ we show the existence of a single-source $p$-multipath preserver of optimal size $p(n-1)$ and in Section $\ref{sec:algorithm}$ we design an algorithm that solves both our problems in time $O(pn^2(p+\log n))$.

\section{An optimal-size single-source \texorpdfstring{$p$}{p}-multipath preserver}
\label{sec:size}
In this section we show that it is possible to compute a single-source $p$-multipath preserver having size $p(n-1)$. 

\begin{figure}
    \centering
    \includegraphics[width=\textwidth]{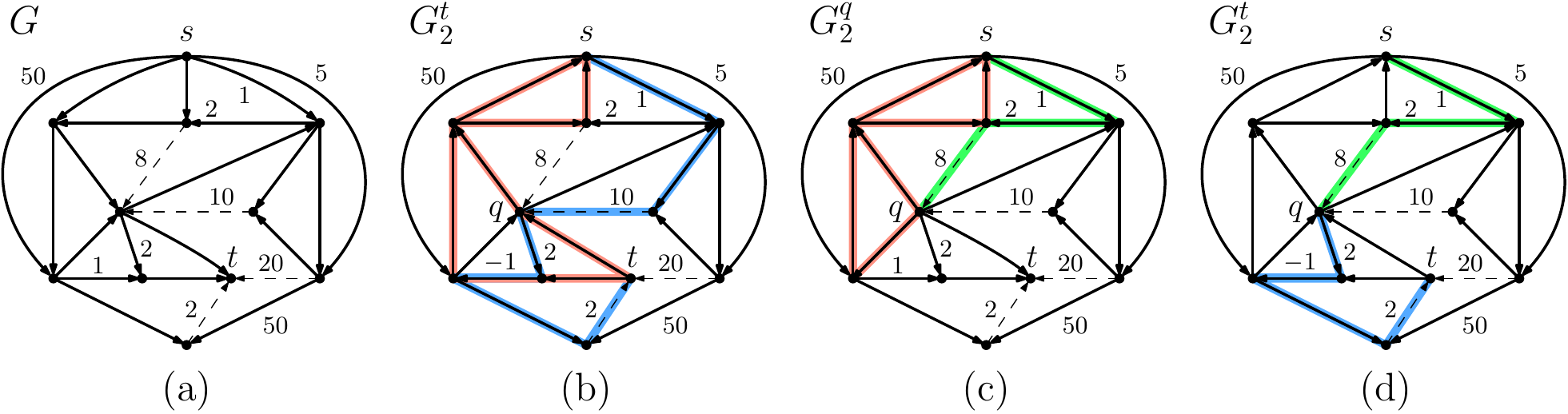}
    \caption{\small An example of the suboptimality property of Lemma~\ref{lemma:suboptimality} for $i=3$.
    (a) The graph $G$, in which the edges in $H_2$ are solid and the edges in $E(G) \setminus E(H_2)$ are dashed. Unlabelled edges cost $0$. For graphical convenience, $G$ is $2$-edge-outconnected from $s$ and can be made $3$-edge-outconnected by a suitable addition of costly edges. (b) The graph $G^t_2$ in which the edges that appear in $S_2^t$ in their opposite orientation are highlighted in red. A shortest path $\pi_t \in \Pi(s,t; G_2^t)$ is highlighted in blue and traverses vertex $q = q(\pi_t)$. (c) The graph $G_2^q$ where the (reversed versions of) the edges in $S_2^q$ are highlighted in red, and a shortest path $\pi_q \in \Pi(s,q; G_2^q)$ is highlighted in green. (d) The graph $G_2^t$ where the path $\pi_q \circ \pi_t[q:t]$ is highlighted in green and blue and belongs to $\Pi(s,t; G_2^t)$.}
    \label{fig:suboptimality_composition}
\end{figure}

We compute such a preserver iteratively: we start with an empty graph $H_0 = (V(G), \emptyset)$ and, during the $i$-th iteration, we construct a $i$-multipath preserver\footnote{In the following we might shorten \emph{single-source $i$-multipath preserver} to \emph{$i$-multipath preserver} or, when $i$ is clear from the context, simply \emph{preserver}.} $H_i$ of $G$ by adding to $H_{i-1}$ a single new edge $e_t$ entering in $t$ for each vertex $t \in V(G) \setminus \{s\}$.
This process stops at the end of the $p$-th iteration. We will show that $H_p$ is a sparse single-source  $p$-multipath preserver. Notice that, by construction, vertex $s$ has in-degree $0$ in $H_i$ and each other vertex has in-degree $i$, therefore $H_p$ has size $p(n-1)$.

We will prove by induction on $i$ that $H_i$ contains $i$ edge-disjoint paths of minimum total cost (in $G$) from $s$ to all vertices $t \in V(G) \setminus \{s\}$.
Since this is trivially true when $i=0$, in the rest of the section we 
assume that the induction hypothesis is true for $H_{i-1}$ with $1 \le i < p$,
and we focus on proving that it remains true for $H_i$.

Following the notation of Section~\ref{sec:successive_shortest_path},
we denote by $S_{i-1}^t$ the set of edges belonging to any $i-1$ edge disjoint paths from $s$ to $t$ of minimum total cost in $H_{i-1}$ (and hence in $G$).
We let $G_{i-1}^t$ be the residual network obtained from $G$ by reversing the edges in $S_{i-i}^t$.

It will be convenient to define distances as pairs of elements from $\mathbb{R} \cup \{+\infty\}$. Given $d = (d_1, d_2)$ and $d' = (d'_1, d'_2)$ we denote by $d+d'$ the pair $(d_1 + d'_1, d_2 + d'_2)$. We also compare distances lexicographically, and write $d \prec d'$ to denote that the pair $d$ precedes $d'$ in the lexicographical order. Similarly, $d \preceq d'$ if $d \prec d'$ or $d = d'$.
Given any path $\pi$, let $\eta(\pi)$ be the number of edges of $\pi$ that are in $E(G)\setminus E(H_{i-1})$.
We can associate $\pi$  with a pair $|\pi| = (c(\pi), \eta(\pi))$. With a slight abuse of notation, we can therefore extend the above linear order 
to paths: for two paths $\pi$ and $\pi'$, we write $\pi \prec \pi'$ (resp. $\pi \preceq \pi'$) as a shorthand for $|\pi| \prec |\pi'|$ (resp. $|\pi| \preceq |\pi'|$).
Intuitively, when we compare paths w.r.t.\ $\preceq$, the values of $\eta(\cdot)$ serve as tie-breakers between paths having the same cost.
In the following $\Pi(u, v; G')$ will denote the set of paths from $u$ to $v$ in $G'$ that are shortest w.r.t.\ the total order relation $\preceq$. When $\Pi(u, v; G')$ contains a single path we denote by $\pi(u, v; G')$ the sole path in $\Pi(u, v; G')$.
Given a path $\pi_1$ from $v_0$ to $v_1$ and a path $\pi_2$ from $v_1$ to $v_2$, we denote by $\pi_1 \circ \pi_2$ the path from $v_0$ to $v_2$ that is obtained by composing $\pi_1$ and $\pi_2$. Given a path $\pi$ from $v_0$ and $v_1$, and two distinct vertices $u$ and $v$ of $\pi$ such that  $\pi$ traverses $u$ and $v$ in this order, we denote by $\pi[u:v]$ the subpath of $\pi$ from $u$ to $v$.

The edge $e_t$ entering $t$ selected by the algorithm is the last edge of an arbitrarily chosen path in $\Pi(s, t; G_{i-1}^t)$.
For $\pi \in \Pi(s, t; G_{i-1}^t)$, we define $q(\pi)$ as the last \emph{internal} vertex of $\pi$ such that its incoming edge in $\pi$ belongs to $E(G) \setminus E(H_{i-1})$. If no such vertex exists, we let $q(\pi)=s$ (see Figure~\ref{fig:suboptimality_composition}~(b)).

\begin{figure}
    \centering
    \includegraphics{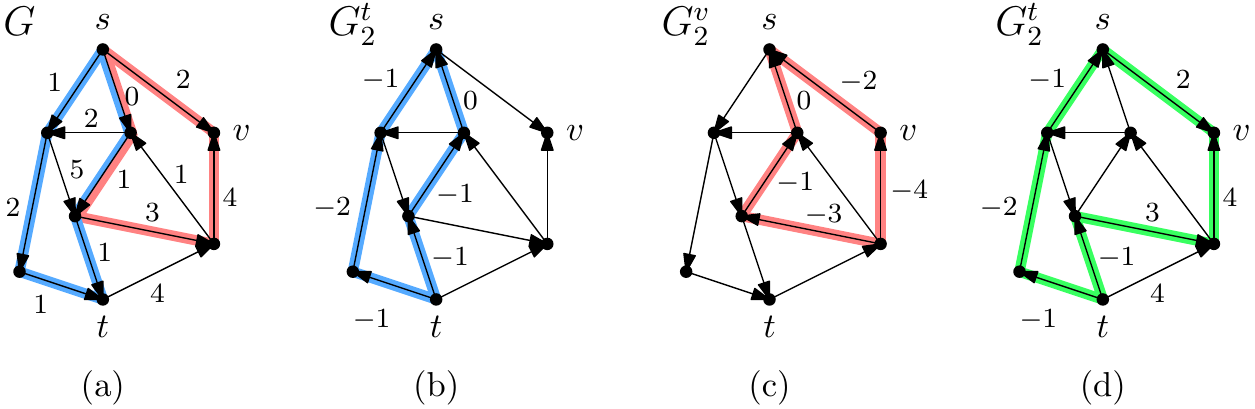}
    \caption{\small (a) A graph $G$ with non-negative costs. The edges in $S_2^t$ (resp. $S_2^v$) are highlighted in blue (resp. red) and induce two edge-disjoint paths of minimum total cost from $s$ to $t$ (resp. $v$) in $G$. 
    The edges in $S^v_2$ are exactly the ones used by the flow $f'$ in the proof of Lemma~\ref{lemma:brown}.
    (b) The graph $G_2^t$ obtained from $G$ by reversing the edges in $S^t_{2}$. The reversed edges (highlighted in blue)  correspond to those used by the flow $f$ in the proof of Lemma~\ref{lemma:brown}. (c) The graph $G_2^v$ obtained from $G$ by reversing the edges in $S^v_{2}$.
The reversed edges are highlighted in red. (d) The graph $G_2^t$ in which the edges in $\Delta(t,v)$ are highlighted in green and induce $2$ edge-disjoint paths of minimum total cost from $t$ to $v$ in $G_2^t$. These edges are the ones used by the flow $f''$ in the proof of Lemma~\ref{lemma:brown}.
The edge costs that are missing in (b), (c), or (d) match those of the corresponding edges in (a). 
}
    \label{fig:delta}
\end{figure}

The main technical ingredient of the result in this section is a suboptimality property, which will be given formally in Lemma~\ref{lemma:suboptimality}.
Intuitively, if $q = q(\pi_t)$ for some path $\pi_t \in \Pi(s,t; G_{i-1}^t)$, then this property ensures that the composition $\pi_q \circ \pi_t[q:t]$ of any shortest path $\pi_q \in \Pi(s,q; G_{i-1}^q)$ with the suffix $\pi_t[q:t]$ of $\pi_t$ is also a shortest path in $\Pi(s,t, G^t_{i-1})$.
Since (up to the orientation of its edges) $\pi_t[q:t]$ contains a single edge $e_t$ not already in $H_{i-1}$ (i.e., the one entering in $t$), this property allows to reuse the edges in $H_{i-1}$ and in $S_i^q \setminus E(H_{i-1})$ to build $S_i^t \setminus \{ e_t \}$.  See Figure \ref{fig:suboptimality_composition}~(d) for an example.
The rest of this section formalizes the above intuition. 

Given any two nodes $t,v \in V(G)$, we denote with $\Delta(t,v)$ the set of edges of $G_{i-1}^t$ that appear in the opposite orientation in $G_{i-1}^v$ (see Figure~\ref{fig:delta}.). Formally, $(x,y) \in \Delta(t,v)$ iff $(x,y) \in E(G_{i-1}^t)$ and $(y,x) \in E(G_{i-1}^v)$.
Equivalently, $(x,y) \in \Delta(t,v)$ iff exactly one of the following conditions hold: (i) $(y,x) \in S_{i-1}^t$; (ii) $(x,y) \in S_{i-1}^v$.
It follows from the above observation that $(x,y) \in \Delta(t,v)$ iff $(y,x) \in \Delta(v,t)$. Moreover,  $E(G) \cap \Delta(t,v) \subseteq S_{i-1}^v \subseteq E(H_{i-1})$. 
The next three lemmas will be instrumental to prove Lemma~\ref{lemma:suboptimality}.

\begin{lemma}
\label{lemma:brown}
The edges in $\Delta(t,v)$ are exactly those belonging to $i-1$ edge disjoint paths of minimum total cost from $t$ to $v$ in $G_{i-1}^t$.
\end{lemma}
\begin{proof}
Consider $G_{i-1}^t$ as an instance of min-cost flow with unit capacity where we want to send $i-1$ units of flow from $t$ to $v$. We define a first flow assignment $f$ that sends $i-1$ units of flow from $t$ to $s$ in $G_{i-1}^t$ using the edges in $S_{i-1}^t$ in the reverse direction (see Figure~\ref{fig:delta}~(b)).
More precisely, $\forall (x,y) \in E(G_{i-1}^t)$, $f(x,y) = 1$ if $(y,x) \in S_{i-1}^t$, and $f(x,y)=0$ otherwise. Notice that $f$ is a flow of value $|f|=i-1$ in $G_{i-1}^t$ and that the associated residual graph is $G$.  
We now consider a minimum-cost flow $f'$ that pushes $i-1$ units of flow from $s$ to $v$ in $G$ using the edges in $S^v_{i-1}$ (see Figure~\ref{fig:delta}~(a) where the edges used by $f'$ are highlighted in red). In particular, we define  $f'(e) = 1$ if $e \in S_{i-1}^v$, and $f(e)=0$ otherwise.
The residual graph associated with $f'$ (w.r.t.\ $G$) is $G_{i-1}^v$ and, since $f'$ is a minimum-cost flow, $G_{i-1}^v$ does not contain any negative-cost cycle~\cite{erickson_mincost_flow}.

We can obtain a flow $f''$ from $t$ to $v$ in  $G_{i-1}^t$ with $|f''|=i-1$ by composing $f$ and $f'$: we first push $i-1$ units of flow from $t$ to $s$  in $G_{i-1}^t$ according to $f$ and then push  $i-1$ units of flow from $s$ to $t$ in the residual network $G$ according to $f'$ (see Figure~\ref{fig:delta}~(d)).
More precisely, the resulting net flow $f''$ is defined as follows: given $(x,y) \in E(G_{i-1}^t)$, $f''(x,y) =1$ iff either (i) $f(x,y)=1$ and $f'(y,x)=0$, or (ii) $f(x,y)=0$ and $f'(x,y)=1$. 
The residual network associated with $f''$ (w.r.t.\ $G_{i-1}^t$) is exactly $G_{i-1}^v$ and, since it contains no negative-cost cycles, $f''$ is also a minimum-cost flow.

To conclude the proof it suffices to notice that the edges $(x,y)$ for which $f''(x,y)=1$ are exactly those in $\Delta(t,v)$.
\end{proof}

\begin{figure}[t]
    \centering
    \includegraphics[scale=1]{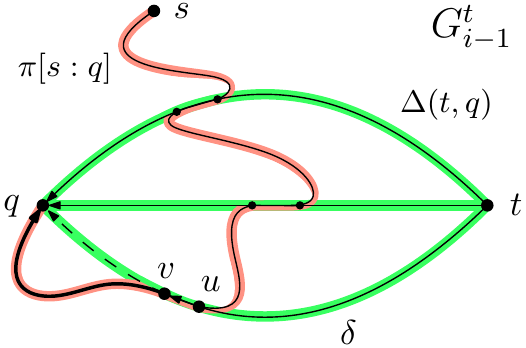}
    \caption{\small A qualitative representation of the proof of Lemma~\ref{lemma:feasible}. We are supposing towards a contradiction that $\pi[s:q]$ (highlighted in red) is not entirely contained in $G_{i-1}^q$ and hence traverses an edge $(u,v)$ belonging to the set $\Delta(t,q)$ (highlighted in green).
    The subpath of $\pi$ (resp. $\delta$)  from $v$ to $q$ is shown in bold (resp. is dashed).}
    \label{fig:suboptimality_containment}
\end{figure}

\begin{lemma}
\label{lemma:feasible}
For every $t \in V(G)$,  let $\pi \in \Pi(s,t;G_{i-1}^t)$ and $q=q(\pi)$.  The subpath $\pi[s:q]$ is entirely contained in $G_{i-1}^{q}$.
\end{lemma}
\begin{proof}
If $s=q$ the subpath $\pi[s:q]$ is empty and the claim is trivially true. We therefore consider $s \neq q$ and suppose towards a contradiction that $\pi[s:q]$ is not entirely contained in $G_{i-1}^q$.
Then, $\pi[s:q]$ traverses at least one edge in $\Delta(t,q)$. Let $(u,v)$ be the last edge traversed by $\pi[s:q]$ that belongs to $\Delta(t,q)$.
By Lemma~\ref{lemma:brown}, the edges in $\Delta(t,q)$ induce  $i-1$ edge disjoints paths of minimum total cost from $t$ to $q$ in the subgraph of $G_{i-1}^t$. Let $\delta$ one such such path traversing $(u,v)$.

Since, by definition of $q$, the edge $e$ of $\pi[s:q]$ entering in $q$ is in $E(G) \setminus E(H_{i-1})$, we have $e \not\in \Delta(t,q)$.
Then, the subpath $\pi[v:q]$ of $\pi[s:q]$ is not empty and, by our choice of $(u,v)$ does not traverse any edge in $\Delta(t,q)$.

By the suboptimality property of shortest paths, $\pi[v:q] \preceq \delta[v:q]$ and hence $c(\pi[v:q]) \leq c(\delta[v:q])$. If $c(\pi[v:q]) < c(\delta[v:q])$, we can replace $\delta[v:q]$ with $\pi[v:q]$ in $\delta$ to obtain a path $\delta'$ from $t$ to $q$ in $G_{i-1}^t$ with $c(\delta') < c(\delta)$.
This contradicts Lemma~\ref{lemma:brown} since it implies the existence of $i-1$ edge-disjoint paths from $t$ to $q$ in $G_{i-1}^t$ with a total cost smaller than $c(\Delta(t,q))$ (see Figure~\ref{fig:suboptimality_containment}).

If $c(\pi[v:q]) = c(\delta[v:q])$, we can replace $\pi[v:q]$ with $\delta[v:q]$ in $\pi[s:q]$ to obtain a path $\pi'$ from $s$ to $q$ in $G_{i-1}^t$ satisfying $c(\pi') = c(\pi[s:q])$. 
Since all edges of $E(G) \cap E(\delta[v:q]) \subseteq \Delta(t,v)$ are in $H_{i-1}$, $\pi[v:q]$ contains more edges in $E(G) \setminus E(H_{i-1})$ than $\delta[v:q]$, thus $\pi' \prec \pi[s:q]$.
This is a contradiction since, by the suboptimality property of shortest paths and by our choice of $\pi \in \Pi(s,t; G_{i-1}^t)$, $\pi[s:q]$ must be a shortest path from $s$ to $q$ in $G_{i-1}^t$ w.r.t.\ $\preceq$.
\end{proof}

\begin{lemma}
\label{lemma:feasible2}
Let $t,q \in V(G) \setminus \{s\}$, and let $\pi$ be a simple path from $s$ to $q$ in $G_{i-1}^q$ such that the edge of $\pi$ entering in $q$ is in $E(G) \setminus E(H_{i-1})$. If $\pi$ is not entirely contained in $G_{i-1}^t$, then there exists a path $\pi'$ from $s$ to $q$ in $G_{i-1}^t$ such that $\pi' \prec \pi$.
\end{lemma}
\begin{proof}
If $\pi$ is not entirely contained in $G_{i-1}^t$ then $\pi$ traverses some edge in $\Delta(q,t)$. Consider the first edge $(u,v) \in \Delta(q,t)$ traversed by $\pi$, and let $\delta$ be a simple path, from $q$ to $t$ that traverses $(u,v)$ in the subgraph of $G_{i-1}^q$ induced by $\Delta(q,t)$.
Since in $G_{i-1}^q$ there are no negative cycles~\cite{erickson_mincost_flow}, we have that $c(\pi[u:q]) + c( \delta[q:u])  \geq  0$ and hence $c(\pi[u:q]) \ge - c( \delta[q:u])$. By reversing the edges in the subpath $\delta[q:u]$ we obtain a path $\delta'$ from $u$ to $q$ that uses only edges in  $\Delta(t,q)$ and has cost $c(\delta') = - c(\delta(u,q))$.
We can then select $\pi' = \pi[s:u] \circ \delta'$. Notice indeed that $c(\pi') = c(\pi[s:u])+ c(\delta') = c(\pi[s:u]) - c(\delta[q:u]) \le c(\pi[s:u]) + c(\pi[u:q]) = c(\pi)$ (see Figure~\ref{fig:alternative_path}). Moreover, $\delta'$ does not use any edge in $E(G) \setminus E(H_{i-1})$ while the last edge in $\pi[u:q]$ is in $E(G) \setminus E(H_{i-1})$. This shows that $\pi' \prec \pi$ and concludes the proof.
\end{proof}

 \begin{lemma}[Suboptimality property]
\label{lemma:suboptimality}
Fix $t \in V(G)$, let $\pi_t \in \Pi(s,t;G_{i-1}^t)$, $q=q(\pi_t)$, and $\pi_q \in \Pi(s, q; G_{i-1}^q)$. We have that $\pi_q \circ  \pi_t[q:t] \in \Pi(s, t; G_{i-1}^t)$.
\end{lemma}
\begin{proof}
We start by showing that $\pi_q$ must  be entirely contained in $G_{i-1}^t$.
To this aim suppose towards a contradiction that $\pi_q$ is not entirely contained in $G_{i-1}^t$.
By Lemma~\ref{lemma:feasible2} there exists a path $\pi'$ in  $G_{i-1}^t$ such that $\pi' \prec \pi_q$ moreover by Lemma~\ref{lemma:feasible}, we know that
$\pi_t[s:q]$ is entirely contained in $G_{i-1}^q$ and since $\pi_q \in \Pi(s,q, G_{i-1}^q)$, we must have $c(\pi_q)  \preceq c(\pi_t[s:q])$.
Thus, we can replace $\pi_t[s,q]$ with $\pi'$ in $\pi_t$ and obtain a new path $\pi'' \prec \pi_t$ from $s$ to $t$ in $G_{i-1}^t$, contradicting $\pi_t \in \Pi(s,t, G_{i-1}^t)$.

Then, the path $\pi = \pi_q \circ \pi_t[q:t]$ obtained by replacing $\pi_t[s:q]$ with $\pi_q$ in $\pi_t$ is entirely contained in $G_{i-1}^t$ and must satisfy $\pi \preceq \pi_t$.
Since $\pi_t$ is a shortest path in $G_{i-1}^t$ w.r.t. $\preceq$, so is $\pi$.
\end{proof}

\begin{figure}
    \centering
    \includegraphics[scale=1]{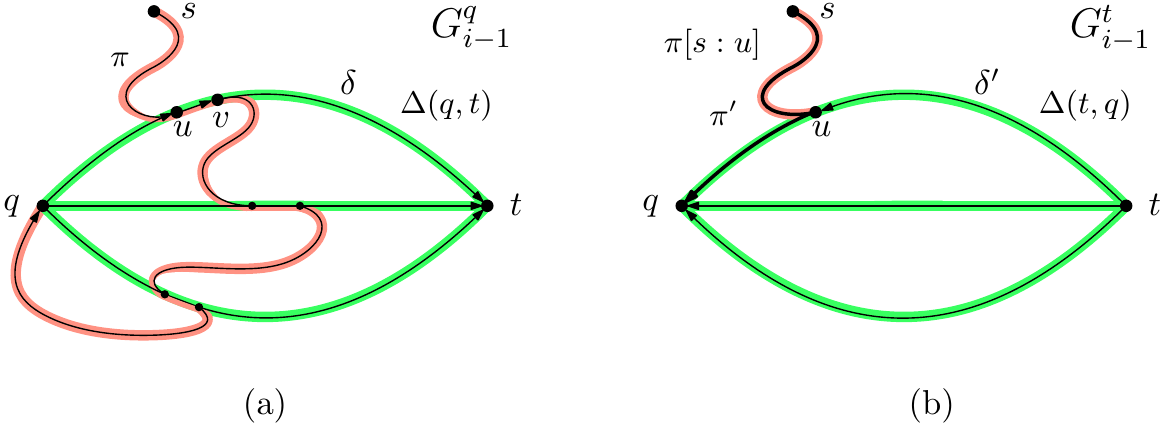}
    \caption{\small A qualitative representation of the proof of Lemma \ref{lemma:feasible2}. We are assuming that $\pi$ (highlighted in red) is a path from $s$ to $q$ in $G_{i-1}^q$ entering in $q$ with and edge of $E(G)\setminus E(H_{i-1})$. (a) $\pi$ intersects a path $\delta$ among the $i-1$ edge-disjoint paths induced by the edges in $\Delta(q,t)$ (highlighted in green). (b)
    The edges of path $\delta'$ obtained by reversing the edges in $\delta$ belong to $\Delta(t,q)$ (highlighted in green).
    Then, the path $\pi' = \pi[s:u] \circ \delta'[u:q]$ (shown in bold) is entirely contained in $G_{i-1}^t$ and satisfies $\pi' \prec \pi$. }
    \label{fig:alternative_path}
\end{figure}

Next lemma  uses the suboptimality property to show that, for each $t \neq s$, there exists a shortest path $\delta$ from $s$ to $t$ in $G_{i-1}^t$ such that, when we orient the edges of $\delta$ in the same direction as in $G$, the resulting set of edges is entirely contained in $H_i$.

\begin{lemma}
\label{lemma:shortest_path_in_H_i}
For each $t \in V(G) \setminus \{s\}$, there exists a path $\delta \in \Pi(s,t, G_{i-1}^t)$ such that $E(\delta) \cap E(G) \subseteq E(H_i)$. 
\end{lemma}
\begin{proof}
Define $q_0 = t$.
For $j \ge 0$ and  $q_j \neq s$,
let $\pi_j \in \Pi(s, q_j; G_{i-1}^{q_j})$ be the shortest from $s$ to $q_j$ selected by the algorithm and define $q_{j+1}=q(\pi_j)$ (see Figure~\ref{fig:path_delta}).

 We now show that all $q_j$ are distinct, hence  there exists a $k$ for which $q_k = s$.  By contradiction, consider the smallest index $j' > j$ such that $q(\pi_{j'}) = q_j$. We will construct two paths towards $q_{j'}$ in $G_{i-1}^{q_{j'}}$ that have different lengths, yet they must both be shortest paths, thus providing the sought contradiction.

By Lemma~\ref{lemma:suboptimality}, we know that $\pi_{j+1} \circ \pi_{j}[q_{j+1}:q_j] \in \Pi(s,q_j;G_{i-1}^{q_j})$. We can repetitively apply Lemma~\ref{lemma:suboptimality}, until we get $\delta_j = \pi_{j'} \circ \pi_{j'-1}[q_{j'}:q_{j'-1}] \circ \ldots \circ \pi_{j}[q_{j+1}:q_{j}] \in \Pi(s,q_j;G_{i-1}^{q_j})$. Since $q(\pi_{j'}) = q_j$, by Lemma~\ref{lemma:suboptimality} we have that $\delta_{j'} = \delta_j \circ \pi_{j'}[q_j:q_{j'}] \in \Pi(s,q_{j'};G_{i-1}^{q_{j'}})$. 

Observe that both $\pi_{j'}$ and $\delta_{j'}$ belong to $\Pi(s,q_{j'};G_{i-1}^{q_{j'}})$, hence must have the same length. However $|\delta_{j'}| = |\delta_j| + |\pi_{j'}[q_j:q_{j'}]| = |\pi_{j'}| +|\delta_j[q_{j'}:q_j]| + |\pi_{j'}[q_j:q_{j'}]|$. As consequence, $|\delta_{j'} | \neq |\pi_{j'}|$ since $\eta(\pi_{j'}[q_j:q_{j'}])=1$ and hence $\eta(\delta_{j'})>\eta(\pi_{j'})$.

Define $\delta = \pi_{k-1}[q_k : q_{k-1}] \circ \pi_{k-2}[q_{k-1} : q_{k-2}] \circ \pi_{k-3}[q_{k-2} : q_{k-3}] \circ \ldots \circ \pi_{0}[q_{1} : q_{0}]$.
We prove by reverse induction on $j=k, \dots, 0$
that (i) $\delta[s:q_j]$ is a shortest path from $s$ to $q_j$ in $G_{i-1}^{q_j}$, and (ii) all edges in $E(\delta[s:q_j]) \cap E(G)$ belong to $H_{i-1}$.
The claim is trivially true for $j=k$ since $q_k=s$ and $\delta[s,q_k]$ is the empty path.
For $j<k$, consider the path $\pi_j$ and notice that $q_{j+1}=q(\pi_j)$ by definition.
By induction hypothesis, we have that $\delta[s:q_{j+1}] \in \Pi(s, q_{j+1}, G_{i-1}^{q_{j+1}})$. Then, by Lemma~\ref{lemma:suboptimality}, 
$\delta[s:q_j] = \delta[s:q_{j+1}] \circ \pi_j[q_{j+1}, q_j] \in \Pi(s, q_j, G_{i-1}^{q_j})$, which proves (i).

As far as (ii) is concerned, we only need to argue about $\pi_j[q_{j+1}, q_j]$ since $\delta[s:q_j] = \delta[s:q_{j+1}] \circ \pi_j[q_{j+1}, q_j]$ and, by induction hypothesis, we know that all edges in $E(\delta[s:q_{j+1}]) \cap E(G)$ are in $E(H_i)$.
Let $(u, q_j)$ be the last edge of $\pi_j[q_{j+1}, q_j]$ and notice that, since $(u, q_j)$ is also the last edge of $\pi_j$, our algorithm adds $(u, q_j)$ to $H_{i}$ when $q_j$ is considered.
Moreover, by the choice of $q_{j+1}=q(\pi_j)$, the path $\pi_j[q_{j+1}: u]$ contains no edges in $E(G) \setminus E(H_{i-1})$. This means that $E(\pi_j[q_{j+1}: u]) \cap E(G)$ lies entirely in $H_{i-1}$ and hence in $H_{i}$.
This shows that all edges of $E(\pi_j[q_{j+1}, q_j]) \cap E(G)$ belong to $E(H_{i})$ and proves (ii).
\end{proof}

\begin{figure}
    \centering
    \includegraphics[scale=1.4]{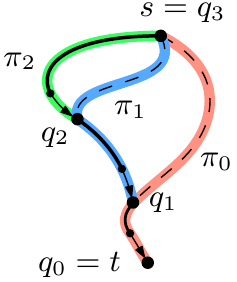}
    \caption{\small A qualitative representation of the path $\delta$ constructed in the proof of Lemma~\ref{lemma:shortest_path_in_H_i} when $k=3$. 
    The path $\delta$ is drawn with solid lines and the portions belonging to $E(H_{i-1})$ are shown in bold.
    The shortest paths $\pi_0$, $\pi_1$, and $\pi_2$ from $s$ to $q_0$, $q_1$, and $q_2$ in $G^{q_0}_{i-1}$, $G^{q_1}_{i-1}$, and $G^{q_2}_{i-1}$ are highlighted in red, blue, and green, respectively. The last edge of each $\pi_j$ (which belongs to $E(G) \setminus E(H_{i-1})$) is drawn as a solid thin line.}
    \label{fig:path_delta}
\end{figure}

\noindent The above lemma easily implies that $H_i$ is a $i$-multipath preserver of $G$. 

\begin{lemma}\label{lemma:H_i_is_spanner}
$H_{i}$ contains $i$ edge-disjoint paths of minimum total cost from $s$ to every $t \in V(G) \setminus \{s\}$.
\end{lemma}
\begin{proof}
    Fix a vertex $t \in V(G) \setminus \{s\}$. By induction hypothesis all edges in $S_{i-1}^t$ belong to $H_{i-1}$. By Lemma~\ref{lemma:shortest_path_in_H_i}, there is a path  $\delta \in \Pi(s,t, G_{i-1}^t)$ such that $E(\delta) \cap E(G) \subseteq E(H_i)$. We now use Remark \ref{remark:oplus} to build $S_{i}^t$ from $S_{i-1}^t$ and $\delta$. It is easy to see that $S_i^t$ must be entirely contained in $H_i$. 
\end{proof}

The combination of Lemma~\ref{lemma:H_i_is_spanner} with the discussion on the size of $H_p$ at the beginning of Section~\ref{sec:size}, immediately results in the following theorem.

\begin{theorem}\label{thm:size_spanner}
$H_p$ is a single-source $p$-multipath preserver of size $p(n-1)$. More precisely, $s$ has in-degree $0$ in $H_p$ while each other vertex has in-degree $p$.
\end{theorem}

\section{An efficient algorithm for finding \texorpdfstring{$p$}{p} edge-disjoint shortest paths}
\label{sec:algorithm}
\begin{algorithm}[t] \small
   	\caption{\small Computes a $p$-multipath preserver of a graph $G$ and $p$-edge disjoint paths of minimum total cost from $s$ to $t$, for every $t \in V(G) \setminus \{s\}$.}
  	\label{alg:1}
  	
    \SetKwInOut{Input}{Input}
    \SetKwInOut{Output}{Output}

    \Input{A graph $G=(V(G),E(G))$, a source vertex $s \in V(G)$, $p \in \mathbb{N}^+$;}
    \Output{$p$ edge-disjoint paths $S_p^t$ from $s$ to $t$ of minimum total cost, $\forall t \in V(G)$;}
    \Output{a $p$-multipath preserver $H_p$ of $G$ with source $s$;}
  	
  	\BlankLine
  	
  	$H_1 \gets $ shortest path tree of $G$ rooted at $s$\label{ln:h1}\tcp*{$H_1$ is a $1$-multipath preserver}
  	\lForEach{$t \in V(G)$}{$S_{1}^t \gets$  path from $s$ to $t$ in $H_1$\label{ln:path_h1}}
    
  	\BlankLine
  	\For(\tcp*[f]{Compute $H_i$ and all $S_i^t$}){$i \gets 2, \dots, p$}
  	{
  	    \ForEach{$t \in V(G) \setminus \{s\}$\label{ln:loop_T_start}}
  	    {  
  	      	$H_{i-1}^t \gets $ Graph obtained from $H_{i-1}$ by reversing the edges in $S_{i-1}^t$ and adding the edges incident to $t$ in $E(G) \setminus S_{i-1}^t$ \label{ln:h_t}\;

  	    $T^{t}_{i-1} \gets $ reverse SPT towards $t$ in $H_{i-1}^t$ \label{ln:reweighted_dijkstra}\label{ln:loop_T_end}\tcp*{$\pi(u,t; T^{t}_{i-1})$ is the sole path in $\Pi(u,t; T^{t}_{i-1})$}
   	    }

        \BlankLine
        \tcp{Initialize distances and priority queue}
        $d(s) \gets (0, 0)$; \quad $\pi_s \gets $ Empty path\label{ln:main_dijkstra_start}\;
        \lForEach{$t \in V(G) \setminus \{s\}$}
        {$d(t) \gets (+\infty, +\infty)$}
        $Q \gets $ initialize a priority queue with values in $V(G)$ and keys $d(\cdot)$\;

        \BlankLine
  	    
  	    $H_i \gets H_{i-1}$\;
  	  	\While{$Q$ is not empty\label{ln:while_begin}}
  	  	{
            $q \gets $ Extract the minimum from $Q$\;
            
            \BlankLine
            
            \If{$q \neq s$}
            {
                $\pi_q \gets \pi_{\rho(q)} \circ \pi(\rho(q), q; T^q_{i-1})$\label{ln:update_path_q}\tcp*{$\pi_q \in \Pi(s, q; G_{i-1}^q)$}
                $e_q \gets$ last edge of $\pi_q$\;
                $E(H_i) \gets E(H_i) \cup \{ e_{q} \}$\label{ln:add_edge_to_H_i}\tcp*{Update the $i$-multipath preserver}
                Compute $S_i^{q}$ from $S_{i-1}^{q}$ and $\pi_{q}$ as explained in Remark \ref{remark:oplus}\label{ln:update_solution_q}
            }
            
            \BlankLine
            
            \ForEach{$t \in Q$}
            {
                \tcp{Check whether $\pi_q \circ \pi(q,t;T^t_{i-1})$ is shorter than $d(t)$}
                \If{$q \in V(T^t_{i-1})$ \, \textbf{\textup{and}} \, $d(q) + |\pi(q,t;T^t_{i-1})| \prec d(t)$}
                {
                    $\rho(t) \gets q$\tcp*{We found a shorter path to $t$ in $G_{i-1}^t$ (via $q$)}
                    $d(t) \gets d(q) + |\pi(q,t;T^t_{i-1})|$\tcp*{Relax $d(t)$}
                    Decrease the key of vertex $t$ in $Q$ to $d(t)$\label{ln:while_end}\;
                }
            }
  	 }  	
  	}
\end{algorithm}

In this section we describe an algorithm (whose pseudocode is given in Algorithm~\ref{alg:1}) running in time $O(p^2n^2+pn^2 \log n)$ that computes: (i) $p$ edge disjoint paths $S_p^t$ of minimum total cost from $s$ to $t$; (ii) a single-source $p$-multipath preserver $H_p$ of size $p(n-1)$ (as stated in Theorem \ref{thm:size_spanner}). Our algorithm also guarantees that each $S_p^t$ is contained in $H_p$.

More precisely, the algorithm will compute along the way all single-source $i$-multipath presevers $H_i$, for $i=1,\dots,p$, as defined in the previous section (recall that $H_i$ has size $i(n-1)$). In this sense, the algorithm can be seen as an efficient implementation of the one described in Section~\ref{sec:size}.

The algorithm works in phases.
The generic $i$-th phase will compute a $i$-multipath preserver $H_i$ from the $(i-1)$-th multipath preserver $H_{i-1}$ computed by the previous phase.
The algorithm also maintains, for each vertex $t$, a solution $S_{i}^t$ consisting of $i$ edge-disjoint paths of minimum total cost from $s$ to $t$. Similarly to $H_i$, $S_{i}^t$ is computed from $S_{i-1}^t$ during phase $i$.

Initially, $H_1$ and all $S_1^t$ are simply a \emph{shortest-path tree} (SPT) of $G$ rooted at $s$, and the (unique) path from $s$ to $t$ in $H_1$.
In each phase $i \ge 2$, the algorithm aims to find a shortest path $\pi_t  \in \Pi(s,t; G^t_{i-1})$. Since a direct computation of $\pi_t$ would be too time-consuming, the idea is that of
exploiting the suboptimality property of Lemma~\ref{lemma:suboptimality}
to consider $\pi_t$ as the composition of two subpaths $\pi_q$ and $\pi_t[q:t]$, where $q = q(\pi_t)$ and $\pi_q$ is a shortest path from $s$ to $q$ in $G_{i-1}^q$.

To this aim, we follow a Dijkstra-like approach (see lines~\ref{ln:main_dijkstra_start}--\ref{ln:while_end}). More precisely, once we have computed $\pi_q$, we attempt to extend it towards every other vertex $t$ by concatenating $\pi_q$ with a shortest path from $q$ to $t$ in $G_{i-1}^t$.

As we have discussed in Section~\ref{sec:successive_shortest_path}, once we have $\pi_t$, we can easily compute $S_{i}^t$ from $S_{i-1}^t$ and $\pi_t$ according to Remark~\ref{remark:oplus}. Moreover, as seen in Section~\ref{sec:size}, we compute $H_i$ by adding to $H_{i-1}$ all the last edges of each $\pi_t$.

There are, however, three caveats that need to be carefully handled.
The first two concerns the algorithm's correctness:
\begin{itemize}
    \item Whenever a path $\pi_q$ is extended towards a vertex $t$, the resulting path may not necessarily exist in $G_{i-1}^t$ (since $\pi_q$ lies in $G_{i-1}^q$ which differs from $G_{i-1}^t$). 
However, this is not an issue since, as we will prove in the following (see Lemma~\ref{lemma:inequality}), when $\pi_q$ does not exist in $G_{i-1}^t$, the length of the resulting path is always an upper bound to the length of $\pi_t$.

    \item In order for the Dijkstra-like approach to work, the vertices $q$ need to be considered in non-decreasing order of $|\pi_q|$, and hence the shortest path from $q$ to $t$ in $G_{i-1}^t$ used to extend $\pi_q$ must have non-negative costs. As we will show, this is indeed the case (see Lemma~\ref{lemma:positive}).
\end{itemize}

The last critical aspect concerns the complexity of the algorithm:
a direct computation of the needed shortest path from $q$ to $t$ in $G_{i-1}^t$ would be too time-consuming.

Instead, we (pre-)compute it in a suitable sparse subgraph of $G_{i-1}^t$, referred as $H_{i-1}^t$ in the pseudocode (see lines~\ref{ln:loop_T_start}--\ref{ln:loop_T_end}). 

\subsection{Proof of correctness}

We prove the correctness of Algorithm~\ref{alg:1} by induction on $i\geq 1$. 
In particular we will shows that, at the end of phase $i$, the following three properties will be satisfied: (i) for $t \in V(G) \setminus \{s\}$, the edges in $S_i^t$ induce $i$ edge-disjoint paths of minimum total cost from $s$ to $t$ in $G$; 
(ii) $H_i$ is a $i$-multipath preserver for $G$ with source $s$; and
(iii) for $t \in V(G) \setminus \{s\}$, $S_i^t$ is entirely contained in $E(H_i)$.

The base case $i=1$ is trivially true since $H_1$ is a shortest path-tree from $s$ in $G$, and $S_i^t$ is the (unique) path from $s$ to $t$ in $H_1$.
We hence assume (i), (ii), and (iii) for $i-1$ and focus on phase $i \ge 2$.

For each $t$, let $G_{i-1}^t$ be the residual network obtained from $G$ by reversing the edges of $S_{i-1}^t$. 
The rest of the proof is organized as follows: we first prove that Algorithm~\ref{alg:1} correctly computes a shortest path $\pi_t \in \Pi(s,t; G_{i-1}^t)$. Then, we will argue that this implies properties (i), (ii), and (iii).

\begin{lemma}
\label{lemma:positive}
Let $t \in V(G)$, and consider the $i$-th phase of Algorithm~\ref{alg:1}. For every $q \in V(T_{i-1}^t)$, we have $c(\pi(q,t;T^t_{i-1})) \ge 0$.
\end{lemma}
\begin{proof}
By contradiction, assume that for some $t,q \in V(G)$, $c(\pi(q,t;T^t_{i-1})) < 0$ . $c(\pi(q,t;T^t_{i-1}))$ is the cost of a shortest path $\pi$ from $q$ to $t$ in $H_{i-1}^t$. If $c(\pi) < 0$, it contains edges that are reversed w.r.t.\ $G$. The set of reversed edges are those belonging to $i-1$ edge disjoint paths from $t$ to $s$ in $G_{i-1}^t$. Let $(x,y)$ the first reversed edge traversed in $\pi$. Consider the subpath $\pi[x:t]$. It holds that $c(\pi[x:t]) \leq c(\pi)$. Consider path $\pi'$ from $t$ to $x$ in $G_{i-1}^t$, that consists in only reversed edges. Thus $\pi[x:t] \circ \pi'$ is a closed walk of negative total cost in $G_{i-1}^t$ that does not contain negative cycles~\cite{erickson_mincost_flow}.
\end{proof}

\begin{lemma}
\label{lemma:inequality}
Let $t \in V(G)$, $\pi \in \Pi(s,t;G_{i-1}^t)$ and consider the $i$-th phase of Algorithm~\ref{alg:1}.
The path $\pi_t$ computed by
line~\ref{ln:update_path_q} after $t$ is extracted from $Q$ satisfies $|\pi| \preceq |\pi_t|$.
\end{lemma}
\begin{proof}
By contradiction, consider first extracted node $t$ for which  $  \pi  \npreceq \pi_t $. 
Let $q=\rho(t)$ be the last node that relaxed node $t$ during phase $i$. Then $\pi_t = \pi_q \circ \pi(q,t;T^t_{i-1})$.

Consider $\pi' \in \Pi(s,q;G_{i-1}^q)$, by hypothesis $\pi' \preceq \pi_q$. There are two cases: (i) $\pi'$ exists in $G_{i-1}^t$, (ii) $\pi'$ does not exists in $G_{i-1}^t$. 
In the first case, $\pi' \circ \pi(q,t;T^t_{i-1})$ exists in $G_{i-1}^t$ then,  $\pi \preceq \pi' \circ \pi(q,t;T^t_{i-1}) \preceq \pi_q \circ \pi(q,t;T^t_{i-1}) = \pi_t$.
In the second case by Lemma \ref{lemma:feasible2}, there exists a path $\pi''$ from $s$ to $q$ in $G_{i-1}^t$ such that $\pi'' \prec \pi'$. Observe that $\pi'' \circ \pi(q,t;T^t_{i-1})$ is an existing path in $G_{i-1}^t$. Then, $\pi \preceq \pi'' \circ \pi(q,t;T^t_{i-1}) \prec \pi' \circ \pi(q,t;T^t_{i-1}) \preceq \pi_q \circ \pi(q,t;T^t_{i-1}) = \pi_t$.
\end{proof}
Since, for each node $t$, the value of $d(t)$ is initialized to $(+\infty, +\infty)$ and it is only decreased during the while loop, the above lemma implies that $d(t)$ is an upper bound on the value $|\pi|$, with $\pi \in \Pi(s,t;G_{i-1}^t)$.

\begin{lemma}
\label{lemma:suitable_subgraph}
Let $t \in V(G)$, $\pi \in \Pi(s,t;G_{i-1}^t)$ and consider the $i$-th phase of Algorithm \ref{alg:1}. The subpath $\pi[q(\pi):t]$ is entirely contained in $H_{i-1}^t$.
\end{lemma}

\begin{lemma}
\label{lemma:inequality2}
Let $t \in V(G)$, $\pi \in \Pi(s,t;G_{i-1}^t)$ and consider the $i$-th phase of Algorithm~\ref{alg:1}.
The path $\pi_t$ computed by
line~\ref{ln:update_path_q} after $t$ is extracted from $Q$ satisfies $|\pi_t| \preceq |\pi|$.
\end{lemma}
\begin{proof}
By contradiction, take the first extracted node $t$ for which  $ \pi_t  \npreceq \pi $. For simplicity let $q = q(\pi)$.

By the suboptimality property (Lemma \ref{lemma:suboptimality}), we have that $\pi = \pi' \circ \pi[q:t]$, where $\pi' \in \Pi(s,q;G_{i-1}^q)$.
By Lemma \ref{lemma:suitable_subgraph} $\pi[q:t]$ exists in $H_{i-1}^t$. Notice that since $H_{i-1}^t \subseteq G_{i-1}^t$, then $|\pi[q:t]| = |\pi(q,t;T^t_{i-1})|$ and by Lemma \ref{lemma:positive}, $c(\pi[q:t]) \geq 0$. Moreover $\pi[q:t]$ contains one edge in $E(G) \setminus E(H_{i-1})$  thus $\pi' \prec \pi$. By hypothesis, $\pi_q \preceq \pi'$ and by Lemma $\ref{lemma:inequality}$ $\pi \preceq \pi_t$, hence $\pi_q \prec \pi_t$ and node $q$ is extracted before $t$. Because of relax check for node $t$ w.r.t.\ $q$ it holds that $\pi_t \preceq \pi_q \circ \pi(q,t;T^t_{i-1}) \preceq \pi' \circ \pi(q,t;T^t_{i-1}) = \pi$.
\end{proof}

\begin{lemma}
\label{lemma:exists}
Let $t \in V(G)$, $\pi \in \Pi(s,t;G_{i-1}^t)$ and consider the $i$-th phase of Algorithm~\ref{alg:1}.
The path $\pi_t$ computed by
line~\ref{ln:update_path_q} after $t$ is extracted from $Q$ is entirely contained in $G_{i-1}^t$.
\end{lemma}
\begin{proof}
By contradiction, take first extracted node $t$ for which  $ \pi_t $ does not exists in $G_{i-1}^t$.
Let $q$ be the last node that performed a relaxation for $t$, we have that $\pi_t = \pi_q \circ \pi(q,t;T^t_{i-1})$.
Since $\pi(q,t;T^t_{i-1})$ exists in $G_{i-1}^t$ then $\pi_q$ does not. By hypothesis $\pi_q$ exists in $G_{i-1}^q$ and by Lemma $\ref{lemma:feasible2}$ there exists a path $\pi'$ in $G_{i-1}^t$ such that $\pi' \prec \pi_q$. The path $ \pi' \circ \pi(q,t;T^t_{i-1})$ gives us an existing path in $G_{i-1}^t$, such that $\pi \preceq \pi' \circ \pi(q,t;T^t_{i-1}) \prec \pi_t$, where $\pi \in \Pi(s,t;G_{i-1}^t)$. This contradicts Lemma \ref{lemma:inequality2} for which $\pi_t \preceq \pi$.
\end{proof}

\noindent We are now ready to establish the correctness of the algorithm as summarized by the following lemma.

\begin{lemma}\label{lm:algorithm_correctness}
For all $i=1,\dots,p$ and $t \in V(G) \setminus \{s\}$, Algorithm~\ref{alg:1} computes a single-source $i$-multipath preserver $H_i$ of $G$ and a set $S^t_i$, with $S_i^t \subseteq E(H)$, inducing $i$ edge-disjoint paths of minimum total cost from $s$ to $t$ in $G$.
\end{lemma}

\subsection{Analysis of the computational complexity}\label{sec:time_complexity_algorithm}

In order to bound the time complexity of our algorithm we first argue on how,
during phase $i$ of Algorithm~\ref{alg:1}, it is possible to
implement Line~\ref{ln:reweighted_dijkstra} in time $O(in + n \log n)$.

For any fixed phase $i$ of the algorithm, and for any target vertex $t$, Line~\ref{ln:reweighted_dijkstra} computes a (reverse) shortest path tree towards $t$ in  $H_{i-1}^t$. As the edge costs in $H_{i-1}^t$ can be negative, a naive implementation using the Bellman-Ford algorithm would require $\Theta(in^2)$ time (since $H_{i-1}^t$ has size $\Theta(in)$). Consequently, the overall time needed to compute all trees $T_{i-1}^t$, for every $i$ and every $t$, would be $\Theta(p^2n^3)$.

To reduce the time complexity of this step, we use a technique similar to the one employed in the successive shortest path algorithm: we re-weight the edges of $H_{i-1}^t$ so that (i) shortest paths are preserved, and (ii) all edge costs are non-negative. Then, after such a re-weighting, a SPT towards $t$ in $H_{i-1}^t$ can be found in $O(in + n \log n)$ time using Dijkstra's algorithm.

We will employ a well-known re-weighting scheme in which the edge costs are completely determined by some function $h : V \to \mathbb{R}$ (see, e.g., \cite[Ch 25.3]{cormen}).
Given $h$, the \emph{new cost} $c'(u,v)$ of an edge $(u,v)$ is defined as $c(u,v) + h(u) - h(v)$.
Notice that the cost of any path $\pi$ from $x$ to $y$ w.r.t.\ $c'$ is exactly $c(\pi) + h(x) - h(y)$, thus the set of shortest paths w.r.t.\ $c'$ coincides with the corresponding set w.r.t.\ $c$.
Therefore, the above re-weighting scheme immediately satisfies (i), and hence we will only need to argue about (ii).

Suppose that, at the beginning of phase $i$ (where $i$ ranges from $2$ to $p$), we already know a re-weighting function $h_{i-2}^t$ such that the graph $G_{i-2}^t$ re-weighted according to $h_{i-2}^t$ has no negative-cost edges.\footnote{Observe that, in the first phase $i=2$, such a function $h_0^t$ is trivially known. Indeed, since the edge costs of $G$ are already non-negative we can simply choose $h_0^t(v) = 0$ for each $v \in V(G)$.}
We will show how to use $h_{i-2}^t$ to obtain a new re-weighting function $h_{i-1}^t$ such that the graph $G_{i-1}^t$ re-weighted according to $h_{i-1}^t$ has no negative-cost edges. 
Since $H_{i-1}^t$ is a subgraph of $G_{i-1}^t$, $h_{i-1}^t$ also satisfies (ii).
The re-weighting induced by $h_{i-1}^t$ can then be immediately used to implement Line~\ref{ln:reweighted_dijkstra} of the algorithm in time $O(in + n \log n)$ using Dijkstra's algorithm.

Let $\bar{H}$ be a graph obtained from $H_{i-1}$ by reversing the edges in $S_{i-2}^t$ (i.e., $i-2$ edge-disjoint paths of minimum total cost from $s$ to $t$), and
notice that $\bar{H}$ is a subgraph of $G_{i-2}^t$.
We compute all the distances from $s$ to the vertices in $\bar{H}$ using Dijkstra algorithm (where edges are re-weighted w.r.t.\ $h_{i-2}^t$) and we let $h(v)$ be the distance from $s$ to $v$. Finally, for each $v \in V(G)$, we define $h_{i-1}^t(v) = h_{i-2}^t(v) + h(v)$.
Lemma~\ref{lemma:h_positive_edge_weights} in the following proves that, when $G_{i-1}^t$ is re-weighted according to $h_{i-1}^t$, all edge costs will be non-negative.

\begin{lemma}
\label{lemma:h_positive_edge_weights}
For any $(u,v) \in E(G_{i-1}^t)$, we have $c(u,v)+h_{i-1}^t(u)-h_{i-1}^t(v) \ge 0$.
\end{lemma}

We conclude by observing that $\bar{H}$ can be computed in $O(in)$ and therefore the overall running time required to compute $T_i^t$ is $O(in + n \log n)$, as claimed. The overall running time of Algorithm~\ref{alg:1} is $O(p^2n^2+pn^2 \log n)$ (see Appendix~\ref{apx:running_time} for further details).
The next theorem follows from Theorem~\ref{thm:size_spanner}, Lemma~\ref{lm:algorithm_correctness}, and the above discussion.

\begin{theorem}\label{thm:main_algorithm}
Algorithm~\ref{alg:1} solves both single-source $p$-multipath preserver problem and shortest $p$ edge-disjoint paths problem in $O(p^2n^2+pn^2 \log n)$. Moreover, the size of the computed preserver is equal to $p(n-1)$, which is optimal.
\end{theorem}

\section{Extensions and variants}
\label{sec:extensions}

In this section we show how to extend all our results to more general versions of the two problems in which the input graph $G$ is not necessarily $p$-edge-outconnected from $s$. More precisely, we denote by $\lambda(t)$ the maximum number of edge-disjoint paths from $s$ to $t$ in $G$. We want to find, for each vertex $t \in V(G)\setminus\{s\}$,  $\sigma(t):=\min\{\lambda(t),p\}$ edge-disjoint paths of minimum total cost from $s$ to $t$ in $G$. 

Moreover, we show how to use Algorithm~\ref{alg:1} to approximate the problem of computing a set of $p$ edge-disjoint paths where the cost of the path with maximum cost is minimized. We show that the approximation factor achieved by our algorithm is optimal.

Finally, for the sake of completeness, we also describe the graph transformation already discussed in~\cite{DBLP:journals/networks/SuurballeT84} if we are interested in finding paths that are vertex-disjoint rather than edge-disjoint.

\subparagraph{Extensions to general versions of our problems.}
W.l.o.g., we can assume that $p < n$ as $G$ contains at most $n-1$ edge-disjoint paths from $s$ to any vertex $t \in V(G)\setminus\{s\}$.

We transform the input graph $G$ into another graph $G'$ that is $p$-edge-outconnected from $s$.
To construct $G'$, we take a copy of $G$ and augment it by adding a complete directed graph $C$ on $p$ new ``dummy vertices'' $v_1, v_2, \dots, v_p$, all edges in $\{ s\} \times \{ v_1, v_2, \dots, v_p \}$, and all edges in $\{v_1, v_2, \dots, v_p\} \times (V(G) \setminus \{s\})$, where the cost of all the new edges is some large value $M > c(E(G))$. We observe that each edge of cost $M$ is incident to at least one dummy vertex. Furthermore, there are $p$ edge-disjoint paths from $s$ to any other vertex of the vertex of the graph, so $H'$ is $p$-edge-outconnected from $s$. As $p < n$, the graph $G'$ still contains $O(n)$ vertices.
We run Algorithm~\ref{alg:1} on $G'$ to compute all the sets $S_i^t$, for each $t\in V(G)\setminus s$ and $i\leq p$ (Lemma~\ref{lm:algorithm_correctness}) in $O(p^2n^2 +pn\log n)$ time. The solution to our problem for $t$ is given by $S_{\sigma(t)}^t$, where we can find the value of $\sigma(t)$ as the largest index $i$ for which $c(S_i^t)<M$. 

Concerning the problem of finding a subgraph $H$ of $G$ such that $S^t_{\sigma(t)} \subseteq E(H)$ for every $t \in V(G)\setminus\{s\}$, we first compute a single-source $p$-multipath preserver $H'$ of $G'$ in $O(p^2n^2 + pn \log n)$ time using Algorithm~\ref{alg:1}. The graph $H$ is obtained from $H'$ by deleting all the dummy vertices and, consequently, all the edges (each of cost $M$) that are incident to the dummy vertices. We observe that $H$, being a subgraph of $G$, does not contain edges of cost $M$.

\begin{theorem}
\label{thm:extension_connectivity}
For every $t \in V(G)\setminus\{s\}$, $S_{\sigma(t)}^t \subseteq E(H)$. Moreover, the size of $H$ is equal to $\sum_{t \in V(G)\setminus\{s\}}\sigma(t)$, which is optimal.
\end{theorem}

\subparagraph{Computing edge-disjoint paths with minimum maximum cost.}

We now consider a variant of the shortest $p$ edge-disjoint paths problem in which we have a different objective function: 
we want to find, for a given source vertex $s$  and every $t \in V(G)\setminus \{s\}$, $p$ edge-disjoint paths from $s$ to $t$ such that the cost of the path with maximum cost is minimized. 
More formally, we want to find, for each $t\in V(G)\setminus \{s\}$, a set $\bar{S}^t_p$ of $p$ edge-disjoint paths from $s$ to $t$ that minimize $\max_{\pi\in \bar{S}^t_p} c(\pi)$. We call this problem the minimum bottleneck $p$ edge-disjoint paths problem.

We observe that, for each $t$, the paths induced by a solution $S^t_p$ for the shortest $p$ edge-disjoint path problem guarantees an approximation factor of $p$. Indeed, $c(S^t_p)\le \sum_{\pi\in \bar{S}^t_p} c(\pi)\le p \cdot \max_{\pi\in \bar{S}^t_p} c(\pi)$. In the next theorem we show that this approximation factor is optimal, unless $\mathrm{P}=\mathrm{NP}$.

\begin{theorem}
\label{thm:max}
There is no polynomial-time algorithm that approximates the minimum bottleneck $p$ edge-disjoint paths problem to within a factor smaller that $p$, unless $\mathrm{P}=\mathrm{NP}$.
\end{theorem}

\subparagraph{Vertex-disjoint paths.}
As also shown by Suurballe and Tarjan~\cite{DBLP:journals/networks/SuurballeT84}, all our results can be extended to the case in which the $p$ paths of minimum total cost from $s$ to $t \in V(G)\setminus\{s\}$ must be pairwise vertex-disjoint via the following linear time reduction. We construct a graph $G'$ by replacing each vertex $v\in V(G)$ with a pair of vertices $v^-,v^+$ that are connected through an edge $(v^-,v^+)$ with cost $0$, and by adding to $E(G')$ an edge $(u^+,v^-)$  of cost $c(u,v)$ for each edge $(u,v)\in E(G)$. We observe that $G'$ still has $O(n)$ vertices. Although $G'$ may not be $p$-edge-outconnected from $s$ (the in-degree of each vertex $v^+$ is equal to 1), we can solve the problem in $O(p^2n^2+pn \log n)$ time using Algorithm~\ref{alg:1}, via the graph transformation that adds $p$ dummy vertices, as described in the previous paragraph.

\bibliographystyle{plainurl}
\bibliography{bibliography}

\clearpage
\appendix

\section{Proofs omitted from Section~\ref{sec:algorithm}}

\begin{proof}[Proof of Lemma~\ref{lemma:suitable_subgraph}]
 Recall that, $H_{i-1}^t$ is defined as the graph obtained from $H_{i-1}$ where the edges in $S_{i-1}^t$ are reversed and contains all the edges $E(G) \setminus S_{i-1}^t$ entering in $t$.
 
 By definition, $\pi[q(\pi):t]$ consists in a sequence of edges (possibly reversed) in $H_{i-1}$ and one edge in  $E(G) \setminus E(H_{i-1})$ entering in $t$. Since by inductive hypothesis $S_{i-1}^t \subseteq E(H_{i-1})$ and  $S_{i-1}^t$ is the set of reversed edges in $G_{i-1}^t$, then $\pi[q(\pi):t]$ exists in $H_{i-1}^t$.
\end{proof}

\begin{proof}[Proof of Lemma~\ref{lm:algorithm_correctness}]
Consider a vertex $t \in V(G) \setminus \{s\}$. By Lemma~\ref{lemma:inequality2} and Lemma~\ref{lemma:exists}, we have that
at the end of phase $i$, $\pi_t \in \Pi(s,t;G_{i-1}^t)$.
Then, by the inductive hypothesis and by Remark~\ref{remark:oplus}, the set $S_i^t$ contains $i$-edge disjoint paths from $s$ to $t$ of minimum total cost in $G$, as desired.
Moreover, since phase $i$ constructs $H_i$ by augmenting $H_{i-1}$ with the last edge of every $\pi_t$, as shown in Section~\ref{sec:size}, we have that $H_i$ is a single-source $i$-multipath preserver of $G$.

    It remains to prove that, at the end of phase $i$ of Algorithm~\ref{alg:1}, all edges in $S_i^t$ are in $H_i$. 

    For any vertex $t$, let $\pi_t$ be the path computed in line~\ref{ln:update_path_q} of Algorithm~\ref{alg:1} during phase $i$.
    Notice that, by construction (see Remark~\ref{remark:oplus}), each edge in $S^t_i$ belongs to at least one of $S_{i-1}^t$ and $E(\pi_t) \cap E(G)$.
    Since we already know that $S^t_{i-1} \subseteq E(H_{i-1}) \subseteq E(H_i)$, we only need to show that $E(\pi_t) \cap E(G) \subseteq E(H_i)$.

    We consider all paths $\pi_t$ computed by Algorithm~\ref{alg:1} during phase $i$, and we prove the above property by induction on the number of edges $\ell$ of $\pi_t$. 

    The base case $\ell=0$ is trivially true since any such path contains no edges.
    Consider now a path $\pi_t$ with $\ell \ge 1$ edges.
    Since $t \neq s$, $\pi_t$ has been computed in line~\ref{ln:update_path_q} as the concatenation of a path $\pi_{\rho(t)}$ with $\pi' = \pi(\rho(t), t; T^t_{i-1})$.
    The path $\pi'$ is entirely contained in $T_{i-1}^t \subseteq H^t_{i-1}$. By construction of $H^t_{i-1}$, the only edges of $G$ that are in $E(H^t_{i-1}) \setminus E(H_i)$ enter $t$.
    Let $e_t = (u,t)$ be the last edge $\pi'$ (this edge always exists since $\rho(t) \neq t$). 
    By the above observation we have that $E(\pi'[s:u]) \cap E(G) \subseteq E(H_{i-1}) \subseteq E(H_i)$, while $e_t$ is added to $H_i$ by line~\ref{ln:add_edge_to_H_i}.
    Finally, $\pi_{\rho(t)}$ satisfies $E(\pi_{\rho(t)}) \cap E(G) \subseteq E(H_i)$  by inductive hypothesis, since $\pi_{\rho(t)}$ has less edges than $\pi_t$.
\end{proof}

\subsection{Proof of Lemma~\ref{lemma:h_positive_edge_weights}}

Before proving Lemma~\ref{lemma:h_positive_edge_weights}, we need a technical lemma showing that the distances from $s$ in $\bar{H}$ and $G_{i-2}^t$ coincide.
 
\begin{lemma}
\label{lemma:spt}
Let $T $ be a shortest path tree rooted at $s$ of $\bar{H}$ then, $T$ is also a shortest path tree rooted at $s$ of $G_{i-2}^t$.
\end{lemma}
\begin{proof}

By contradiction, if $T$ is not a shortest path tree of $G_{i-2}^t$, it is because there exists some node $v$ in $G_{i-2}^t$ for which every path $\pi \in \Pi(s,v;G_{i-2}^t)$ contains some edge from $E(G) \setminus E(\bar{H})$. 
Fix a $\pi \in \Pi(s,v;G_{i-2}^t)$ and assume w.l.o.g.\ that $\pi$ contains only one edge from $E(G) \setminus E(\bar{H})$ and that this edge enters in $v$.

We now show that $\pi(s,v;T) \preceq \pi $. 
Consider $\pi_v \in \Pi(s,v;G_{i-2}^v)$ computed by Algorithm \ref{alg:1} during phase $i-1$, and observe that for each $(u,v) \in \pi_v$ either $(u,v) \in E(H_{i-1})$ or $(v,u) \in E(H_{i-1})$. By Lemma \ref{lemma:feasible2}, either $\pi_v$ exists in $G_{i-2}^t$ or there exists a path $\pi_v'$ in $G_{i-2}^t$, obtained from $\pi_v$ by substituting a subpath in $\pi_v$ with a path containing only edges from $\Delta(t,v)$ w.r.t.\ $G_{i-2}^t$ and $G_{i-2}^v$ and such that $\pi_v' \prec \pi_v$ . Let $\pi' $ be the existing path in $G_{i-2}^t$ between $\pi_v'$ and $\pi_v$. By construction, $\pi'$ is a path that consists only in edges from $\bar{H}$, thus $\pi(s,v;T) \preceq \pi'$. 

To conclude the proof, we need to show that $\pi$ exists in $G_{i-2}^v$. Similarly to the proof of Lemma \ref{lemma:feasible}, if $\pi$ does not exists in $G_{i-2}^v$, it traverses at least one edge in $\Delta(t,v)$ w.r.t.\ $G_{i-2}^t$ and $G_{i-2}^v$. Let $(x,y)$ be the last edge traversed by $\pi$ that belongs in $\Delta(t,v)$.
Let $\delta$ be a path from $t$ to $y$ that traverses $(x,y)$ in the subgraph of $G_{i-2}^t$ induced by the edges in $\Delta(t,v)$.

Since, by definition of $\pi$, the edge $e$ of $\pi$ entering in $v$ is not in $\bar{H}$, we have $e \not\in \Delta(t,v)$.
Then, the subpath $\pi[y:v]$ of $\pi$ is not empty and, by our choice of $(x,y)$ does not traverse any edge in $\Delta(t,v)$.

By the suboptimality property of shortest paths, $\pi[y:v] \preceq \delta[y:v]$ and hence $c(\pi[y:v]) \leq c(\delta[y:v])$. If $c(\pi[y:v]) < c(\delta[y:v])$, we can replace $\delta[y:v]$ with $\pi[y:v]$ in $\delta$ to obtain a path $\delta'$ from $t$ to $v$ in $G_{i-2}^t$ with $c(\delta') < c(\delta)$.
This contradicts Lemma~\ref{lemma:brown} since it implies the existence of $i-2$ edge-disjoint paths from $t$ to $v$ in $G_{i-2}^t$ with a total cost smaller than $c(\Delta(t,v))$.

If $c(\pi[y:v]) = c(\delta[y:v])$, we can replace $\pi[y:v]$ with $\delta[y:v]$ in $\pi$ to obtain a path $\pi''$ from $s$ to $v$ in $G_{i-2}^t$ satisfying $c(\pi'') = c(\pi)$. 
Since all edges (or their reverse) of $E(\delta[y:v]) \subseteq \Delta(t,v)$ are in $H_{i-2}$, $\pi[y:v]$ contains more edges in $E(G) \setminus E(H_{i-2})$ than $\delta[y:v]$, thus $\pi'' \prec \pi$.
This is a contradiction since, by the suboptimality property of shortest paths and by our choice of $\pi \in \Pi(s,v; G_{i-2}^t)$, $\pi$ must be a shortest path from $s$ to $v$ in $G_{i-2}^t$ w.r.t.\ $\preceq$.

Then knowing that $\pi$ exists also in $G_{i-2}^v$, it holds that $\Pi(s,v;T) \preceq \pi'\preceq \pi$.
\end{proof}

\begin{proof}[Proof of Lemma~\ref{lemma:h_positive_edge_weights}]
    Let $c'(u,v)$ denote the cost of edge $(u,v)$ in $G_{i-2}^t$, when the graph is re-weighted according to $h_{i-2}^t$. Notice that, by hypothesis, $c'(u,v)$ is always non-negative. Recall that $h(v)$ is the distance from $s$ to $v$ in $\bar{H}$ w.r.t.\ $c'$ and that, by Lemma~\ref{lemma:spt}, $h(v)$ is also the distance from $s$ to $v$ in $G_{i-2}^t$ w.r.t.\ the cost function $c'$.

Thus we have that, for each edge $(u,v) \in E(G_{i-2}^t)$,  $h(u) + c'(u,v) \ge h(v)$ implying that
$c(u,v) + h_{i-1}^t(u) - h_{i-1}^t(v) = c'(u,v) + h(u) - h(v) \ge 0$.
In particular, if $(u,v) \in E(G_{i-2}^t)$ belongs to a shortest path (w.r.t.\ $c'$) from $s$ to $v$ in $G_{i-2}^t$ then, $h(u) + c'(u,v) = h(v)$ and
$c(u,v) + h_{i-1}^t(u) - h_{i-1}^t(v) = c'(u,v) + h(u) - h(v) = 0$.

By definition, $G_{i-1}^t$ is obtained from $G_{i-2}^t$ by reversing $\pi \in \Pi(s,t;G_{i-2}^t)$.

For each $(u,v) \in E(G_{i-1}^t) \cap E(G_{i-2}^t), c(u,v) + h_{i-1}^t(u) - h_{i-1}^t(v) \geq 0$ and for each $(u,v) \in E(G_{i-1}^t) \setminus E(G_{i-2}^t)$ we have that $(v,u) \in \pi$ then, $ c(u,v)  + h_{i-1}^t(u) - h_{i-1}^t(v) = -(c(v,u) - h_{i-1}^t(u) + h_{i-1}^t(v)) = 0$.
\end{proof}

\subsection{Running time of Algorithm~\ref{alg:1}.}
\label{apx:running_time}

\begin{proof}[Proof of Theorem~\ref{thm:main_algorithm}]
We can ignore lines~\ref{ln:h1} and \ref{ln:path_h1} since they require time $O(n^2)$. We therefore focus on an iteration $i\ge 2$ of the outer loop (i.e., on phase $i$).

The discussion in Section~\ref{sec:time_complexity_algorithm} shows that the loop at lines~\ref{ln:loop_T_start}--\ref{ln:loop_T_end} requires time $O(in^2 + n^2 \log n)=O(pn^2+n^2 \log n)$.
Observe that line~\ref{ln:update_path_q} can be implemented in time proportional to the number of edges of $\pi_q$, which is at most $n-1$, and that line~\ref{ln:update_solution_q} requires time at most $O(|S_{i-1}^q| + n) = O(pn)$.
We implement the priority queue $Q$ using a data structure that supports decrease-key operations in constant-time (e.g., an array). 
Since we perform $O(n)$ extract-min operations, and $O(n^2)$ decrease-key operations, we have that the loop at lines~\ref{ln:main_dijkstra_start}--\ref{ln:while_end} requires time $O(in^2)=O(pn^2)$. 

Thus, the overall time complexity of the algorithm is $O(p^2n^2+pn^2 \log n)$.
\end{proof}

\section{Proofs omitted from Section~\ref{sec:extensions}}

\begin{proof}[Proof of Theorem~\ref{thm:extension_connectivity}]
By the algorithmic construction of the single-source $p$-multipath preserver, we have that $H'=H_p'$ contains $H_i'$, for every $i \leq p$. Since the edges of  $S_{\sigma(t)}^t$ are also edges of $H_{\sigma(t)}'$ (Lemma~\ref{lm:algorithm_correctness}), it follows that $S_{\sigma(t)}^t \subseteq E(H_p')$, and thus $S_{\sigma(t)}^t \subseteq E(H)$, as $S_{\sigma(t)}^t$ has no edge of cost $M$.

The lower bound of $\sum_{t \in V(G)\setminus\{s\}}\sigma(t)$ on the size of any feasible solution to the problem comes from the fact that the in-degree of each vertex $t$ must be at least $\sigma(t)$. 

We now prove that the size of $H$ matches the lower bound by showing that the in-degree of each vertex $t \in V(G)\setminus\{s\}$ equals $\sigma(t)$. Using the fact that the in-degree of $t$ in $H'$ is exactly equal to $p$ (Theorem~\ref{thm:main_algorithm}), it is enough to show that there are $p-\sigma(t)$ edges of cost $M$ that are entering $t$ in $H'$. 

Consider the solution $S^t_p$ that contains $p$ edge-disjoint paths $\pi_1,\dots,\pi_t$ from $s$ to $t$ in $G'$ of minimum total cost. W.l.o.g., we assume that $c(\pi_1)\leq \dots \leq c(\pi_p)$. We claim that for each $i$ with $\sigma(t) < i \leq p$, $\pi_i$ enters $t$ with an edge of cost $M$. To show this it is enough to observe two things. From the one hand,  $c(\pi_i)$ must have a cost of at least $M$ as otherwise we would have $\sigma(t)+1$ edge-disjoint paths from $s$ to $t$ in $G$ of total cost of at most $c(E(G))<M$, thus contradicting the assumption that there are at most $\sigma(t)=\lambda(t)$ edge-disjoint paths from $s$ to $t$ in $G$. On the other hand, every path from $s$ to $t$ in $G'$ of cost of at least $M$ has a cost that is actually lower bounded by $2M$. This is because any such path must pass through a dummy vertex which has only edges of cost $M$ incident to it. As a consequence, $c(\pi_i) \geq 2M$ for every $\sigma(t)< i \leq p$.

To complete the proof, it is enough to notice that each path of cost equal to $2M$ from $s$ to $t$ passes through a single dummy vertex and enters in $t$ with an edge of cost $M$. As there are $p$ dummy vertices, there are also $p$ edge-disjoint paths from $s$ to $t$ of cost $2M$ each. 
This implies that each path $\pi_i$ from $s$ to $t$ of cost strictly larger than $2M$ can be replaced by a path of cost exactly equal to $2M$ using shortcuts (i.e., the direct edge from the first dummy vertex traversed in $\pi$ to $t$). If we do this simultaneously for all the paths $\pi_1,\dots,\pi_p$ of total cost strictly larger than $2M$, we obtain a new set of paths $\pi_1',\dots,\pi_p'$ that are still pairwise edge-disjoint and such that $\sum_{i=1}^p c(\pi_i') < \sum_{i=1}^p c(\pi_i)$. Therefore, by the optimality of $S_p^t$, $c(\pi_i)=2M$ for every $\sigma(t) < i \leq p$. As a consequence, each $\pi_i$, with $\sigma(t) < i \leq p$, enters in $t$ with an edge of cost $M$. Therefore, $p-\sigma(t)$ edges out the $p$ edges entering $t$ in $H'$ are of cost $M$ each. Hence, the degree of $t$ in $H$ is equal to $\sigma(t)$.
\end{proof}

\begin{proof}[Proof of Theorem~\ref{thm:max}]
We prove the statement for a given pair of nodes $s$ and $t$. We reduce from the 2 directed paths problem (2DP): Given a directed graph $G=(V(G),E(G))$ and four vertices $s_1,t_1,s_2,t_2 \in V(G)$, decide if there exist two edge disjoint paths, one from $s_1$ to $t_1$ and one from $s_2$ to $t_2$. The 2DP problem is NP-Complete~\cite{DBLP:journals/tcs/FortuneHW80}.

Starting from the input graph $G$ of 2DP, we define a graph $G'$ which consists in $p-1$ copies  $G_1,\ldots G_{p-1}$ of $G$ and two new vertices $s$ and $t$ and we set the cost of each edge in every copy to $0$. We denote by $s_j^i$ and $t_j^i$, the nodes $s_j$ and $t_j$ in the $i$-th copy of $G$, respectively, for each $1 \leq i\leq p-1$ and $j=1,2$. 
Node $s$ has $p-1$ edges of cost $0$ toward nodes $s^i_2$, for each $1 \leq i\leq p-1$, and one edge of cost $1$ toward node $s^1_1$. Node $t$ has $p-1$ edges of cost $0$ from nodes $t^i_1$, for each $1 \leq i\leq p-1$, and one edge of cost $1$ from $t^{p-1}_2$. Moreover, there is an edge $(t^i_2,s^{i+1}_1)$ of cost $1$, for each $1 \leq i\leq p-2$, see Figure~\ref{fig:max} for an illustration. 

We first show that, if in $G$ there are two edge-disjoint paths, one from $s_1$ to $t_1$ and one from $s_2$ to $t_2$, then in $G'$ there are $p$ edge-disjoint paths from $s$ to $t$ each of them with cost $1$. For each $1 \leq i\leq p-1$, let us denote by $\pi^i_1$ and $\pi^i_2$ the two disjoint paths in $G_i$ from $s_1^i$ to $t_1^i$ and from $s_2^i$ to $t_2^i$, respectively.
The first of the $p$ edge-disjoint paths from $s$ to $t$ in $G'$ starts from $s$, goes to $s^1_1$, follows path $\pi_1^1$ and then goes from $t^1_1$ to $t$. The total cost of this path is $1$.
A second path starts from $s$ goes to $s^{p-1}_2$, follows path $\pi^{p-1}_2$ from $s^{p-1}_2$ to $t^{p-1}_2$ in $G_{p-1}$ and then goes from $t^{p-1}_2$ to $t$. The total cost this path is $1$. The remaining $p-2$ are constructed in this way: Each path $i$, with $1\leq i\leq p-2$, starts from $s$ goes to $s^i_2$, follows path $\pi^{i}_2$ from $s^i_2$ to $t^i_2$ in $G_i$ and then crosses edge $(t^i_2,s^{i+1}_1)$. In $G_{i+1}$, it follows path $\pi_1^{i+1}$ from $s^{i+1}_1$ to $t^{i+1}_1$ and finally crosses edge $(t^{i+1}_1,t)$. The cost of each of these paths is $1$. By construction these $p$ paths are edge-disjoint.

Now we show that, if in $G$ there are not two edge disjoint paths, one from $s_1$ to $t_1$ and one from $s_2$ to $t_2$, then in $G'$ any $p$ edge disjoint paths from $s$ to $t$ contain a path of cost $p$. We can assume w.l.o.g. that there are 2 edge-disjoint paths in $G$, one from $s_1$ to $t_2$ and one from $s_2$ to $t_1$.
In $G'$ there is only one possible set of $p$ edge-disjoint paths, which is made of $p-1$ paths of cost $0$ and one path of cost $p$. The first $p-1$ paths are composed as follows: for $1 \leq i \leq p-1$, each of these paths starts from $s$ and goes to node $s^i_2$ in $G_i$ through edge $(s,s^i_2)$, it follows a path $\pi^i_1$ from $s^i_2$ to $t^{i}_1$ (in $G_i$) disjoint from a path $\pi^i_2$ between $(s^i_1,t^i_2)$ (in $G_i$), and then reaches $t$ by edge $(t^i_1,t)$. The last path starts from $s$ and by crossing edge $(s,s^1_1)$ of cost $1$, follows $\pi^i_2$ to reach node $t^1_2$. At this point, it keeps moving along all copies $G_i$ of $G$ by using edges $(t^i_2,s^{i+1}_1)$ of cost $1$ and by using path $\pi^i_2$ to reach $t^i_2$ from $s^i_1$. Finally, the last edge crossed is $(t^{p-1}_2,t)$. The total cost of this path is $p$.

It follows that an algorithm that approximates the minimum bottleneck $p$ edge-disjoint paths problem to within a factor smaller that $p$ can be used to solve 2DP.
\end{proof}

\begin{figure}[t]
    \centering
    \includegraphics[scale=0.75]{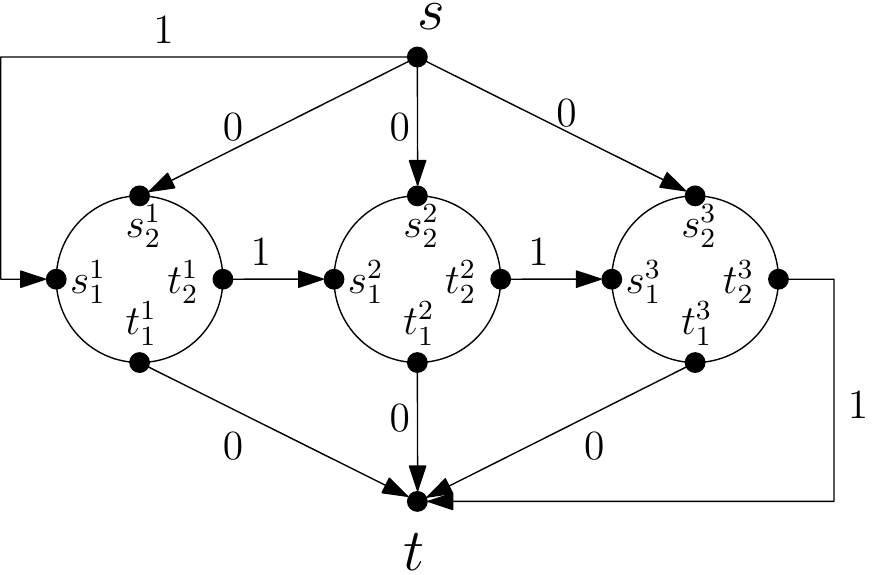}
    \caption{\small Reduction used in Theorem~\ref{thm:max}.}
    \label{fig:max}
\end{figure}

\end{document}